%% file: ijcai20.tex
\documentclass[twocolumn]{article}
\usepackage{fullpage}

\usepackage{times}
\usepackage{soul}
\usepackage{url}
\usepackage[hidelinks,colorlinks=true,linkcolor=black]{hyperref}
\usepackage[utf8]{inputenc}
\usepackage[small]{caption}
\usepackage{graphicx}
\usepackage{amsmath}
\usepackage{amsthm}
\usepackage{booktabs}
\usepackage{algorithm}
\usepackage{algorithmic}
\usepackage{natbib}

\newtheorem{example}{Example}
\newtheorem{theorem}{Theorem}

\title{Maximizing Welfare with Incentive-Aware Evaluation Mechanisms\thanks{Published in IJCAI 2020.}}

\usepackage{authblk}
\author[1]{Nika Haghtalab}
\author[2]{Nicole Immorlica}
\author[2]{Brendan Lucier}
\author[1]{Jack Z. Wang}
\affil[1]{Cornell University, (\textit{nika,jackzwang@cs.cornell.edu})}
\affil[2]{Microsoft Research, (\textit{nicimm,brlucier@microsoft.com})}

\date{}

\input{macros}

\begin{document}
\maketitle

\input{abstract}

\input{intro}
\input{prelim}
\input{linear}

\input{threshold}
\input{learning}

\input{discussion}

\bibliographystyle{plainnat}
\bibliography{Nika_bib}

\onecolumn
\appendix

\input{app_smooth}

\input{app_proof_linear}
\input{app_proof_threshold}

\input{app_fullrank}

\end{document}

%% file: macros.tex
\usepackage{hyperref}
\hypersetup{
     colorlinks   = true,
     urlcolor    = black,
	 citecolor = black
}

\usepackage{amssymb,amssymb,wrapfig}
\usepackage[dvipsnames]{xcolor}
\usepackage{xspace}
\usepackage{nicefrac}
\usepackage{enumitem}
\usepackage{wrapfig}
\newtheorem{lemma}{Lemma}

\newtheorem{prop}{Proposition}
\newtheorem{definition}{Definition}

\newcommand{\comment}[1]{}

\DeclareMathOperator*{\argmax}{argmax\,}
\DeclareMathOperator*{\E}{\mathbb{E}}
\newcommand{\R}{\mathbb{R}}
\newcommand{\X}{\mathcal{X}}
\newcommand{\D}{\mathcal{D}}
\renewcommand{\P}{\mathcal{P}}
\newcommand{\K}{\mathcal{K}}

\newcommand{\poly}{\mathrm{poly}}
\newcommand{\sign}{\mathrm{sign}}
\newcommand{\ind}{\mathbf{1}}
\renewcommand{\vec}[1]{\mathbf{#1}}

\newcommand{\sden}{\mathrm{S}\text{-}\mathrm{Den}}
\newcommand{\mden}{\mathrm{Den}}

\newcommand{\cost}{\mathrm{cost}}
\newcommand{\G}{\mathcal{G}}
\newcommand{\F}{\mathcal{F}}
\newcommand{\Glin}{\mathcal{G}_{\text{lin}}}
\newcommand{\Gth}{\mathcal{G}_{\text{0-1}}}

\newcommand{\Val}{\mathrm{Val}}
\newcommand{\gain}{\mathrm{Gain}}

\newcommand{\opt}{\mathrm{opt}}

\newcommand\numberthis{\addtocounter{equation}{1}\tag{\theequation}}
\allowdisplaybreaks

%% file: abstract.tex
\begin{abstract}
Motivated by applications such as college admission and insurance rate determination, we propose an evaluation problem where the inputs are controlled by strategic individuals who can modify their features at a cost.  A learner can only partially observe the features, and aims to classify individuals with respect to a quality score.
The goal is to design an evaluation mechanism that maximizes the overall quality score, i.e., welfare, in the population, taking any strategic updating into account.

We further study the algorithmic aspect of finding the welfare maximizing evaluation mechanism under two specific settings in our model.
When scores are linear and mechanisms use linear scoring rules on the observable features, we show that the optimal evaluation mechanism is an appropriate projection of the quality score.  When mechanisms must use linear thresholds, we design a polynomial time algorithm with a (1/4)-approximation guarantee  when the underlying feature distribution is sufficiently smooth and admits an oracle for finding dense regions.  We extend our results to settings where the prior distribution is unknown and must be learned from samples.
\end{abstract}

%% file: intro.tex
\section{Introduction}

Automatic decision making is increasingly used to identify qualified individuals in areas such as education, hiring, and public health. This has inspired a line of work aimed at improving the performance and interpretability of classifiers for \emph{identifying qualification and excellence} within a society given access to limited visible attributes of individuals.
As these classifiers become widely deployed at a societal level, they can take on the additional role of \emph{defining excellence and qualification}. That is, classifiers encourage people who seek to be ``identified as qualified'' to acquire attributes that are ``deemed to be qualified'' by the classifier.  For example, a college admission policy that heavily relies on SAT scores 
will naturally encourage students to increase their SAT scores, which they might do by getting a better understanding of the core material, taking SAT prep lessons, cheating, etc.
 Progress in machine learning has not fully leveraged classifiers' role in incentivizing people to change their feature attributes, and at times has even considered it an inconvenience to the designer who
 must now take steps to ensure that their classifier cannot be ``gamed''~
\citep{meir2012algorithms,chen2018strategyproof,hardt2016strategic,dekel2010incentive,cai2015optimum}.  One of the motivations for such \emph{strategy-proof classification} is Goodhart's law, which states ``when a measure becomes a target, it ceases to be a good measure.''  Taking Goodhart's law to heart, one might view an individual's attributes to be immutable, and any strategic response to a classifier (changes in one's attributes)
only serves to mask an agent's true qualification and thereby degrade the usefulness of the classifier.

What this narrative does not address is that in many applications of classification,
one's qualifications can truly be improved by changing one's attributes.
For example, students who improve their understanding of core material  truly become more qualified for college education. These changes have the potential to raise the overall level of qualification and excellence in a society and should be encouraged by a social planner.
 In this work, we focus on this powerful and under-utilized role of classification in machine learning and ask how to
\begin{quote}
\emph{design an evaluation mechanism on visible features that incentivizes individuals to improve a desired quality.}
\end{quote}
For instance, in college admissions, the planner might wish to maximize the ``quality'' of candidates. Quality is a function of many features, such as persistence, creativity, GPA, past achievements, only a few of which may be directly observable by the planner. Nevertheless the planner designs an admission test on visible features to identify qualified  individuals. To pass the test, ideally candidates improve their features and truly increase their quality as a result.
In another motivating example, consider a designer who wants to increase average driver safety, which can depend on many features detailing every aspect of a driver's habits. The designer may only have access to a set of visible features such as age, driver training, or even driving statistics like acceleration/deceleration speed (as recorded by cars' GPS devices).  Then a scoring rule (that can affect a driver's insurance premium) attempts to estimate and abstract a notion of ``safety'' from these features.  Drivers naturally adjust their behavior to maximize their score.  In both cases, the mechanism does not just pick out high quality candidates or safe drivers in the underlying population, but it actually causes their distributions to change.

\paragraph{Modeling and Contributions.}
Motivated by these observations, we introduce the following general problem.
Agents are described by feature vectors in a high-dimensional feature space, and can change their innate features at a cost.  There is a \emph{true} function mapping feature vectors to a true quality score (binary or real-valued).
The planner observes a low-dimensional projection of agents' features and chooses an evaluation mechanism, from a fixed class of mechanisms which map this low-dimensional space to \emph{observed} scores (binary or real-valued).  Agents get value from having a high observed score (e.g.,  getting  admitted to university or having a low car insurance premium), whereas the planner wishes to maximize welfare, i.e., the average true score of the population.

To further study this model, we analyze the algorithmic aspects of maximizing welfare in two specific instantiations.
In Section~\ref{sec:linear-g}, we show that when  the true quality function is linear and the mechanism class is the set of all linear functions, the optimal is a projection of the true quality function on the visible subspace.
Our most interesting algorithmic results (Section~\ref{sec:threshold-g}), arise from the case when the true function is linear and mechanism class is the set of all linear thresholds. 
In this case, a simple projection does not work: we need to consider the distribution of agents (projected on the visible feature space) when choosing the mechanism. For this case, we provide polynomial time approximation algorithms for finding the optimal linear threshold.
In Section~\ref{sec:learn},  we also provide sample complexity guarantees for learning the optimal mechanism from samples only.

\paragraph{Prior Work.}
Our work builds upon the strategic machine learning literature introduce by \citet{hardt2016strategic}.
   As in our work, agents are represented by feature vectors which can be manipulated at a cost.
 \citet{hardt2016strategic} design optimal learning algorithms in the presence of these costly strategic manipulations.  \citet{HuIV19} and \citet{MilliMDH19} extend \citep{hardt2016strategic} by assuming different groups of agents have different costs of manipulation and study the disparate impact on their outcomes.  \citep{DongRSWW18} consider a setting in which the learner does not know the distribution of agents' features or costs but must learn them through revealed preference.  Importantly, in all these works, the manipulations do not change the underlying features of the agent and hence purely disadvantage the learning algorithm. \citet{KleinbergR19}
introduce a different model in which manipulations do change the underlying features.  Some changes are advantageous, and the designer chooses a rule that incentivizes these while discouraging disadvantageous changes.  Their main result is that simple linear mechanisms suffice for a single known agent (i.e., known initial feature vector).  In contrast, we study a population of agents with a known distribution of feature vectors and optimize over the class of linear, or linear threshold, mechanisms.  \citet{alon2019multiagent} also study extensions of \citet{KleinbergR19} to multiple agents. In that work, agents differ in how costly it is for them to manipulate their features but they all have the same starting feature representation, but in our work, agents differ in their starting features while facing the same manipulation cost.

As noted by \citet{KleinbergR19}, their model
is closely related to the field of contract design in economics. The canonical principal-agent model (see, e.g., \citep{GrossmanH83,Ross73}) involves a single principle and a single agent.  There is a single-dimensional output, say the crop yield of a farm, and the principal wishes to incentivize the agent to expend costly effort to maximize output.  However, the principle can only observe output, and the mapping of effort to output is noisy.  Under a variety of conditions, the optimal contract in such settings pays the agent an  function of output~\citep{Carroll15,DuettingRT19,HolstromM87}, although the optimal contract in general settings can be quite complex~\citep{McafeeM86}.  Our model differs from this canonical literature in that both effort and output are multi-dimensional.  In this regard, the work of \citet{HolstromM91} is closest to ours.  They also study a setting with a multi-dimensional feature space in which the principal observes only a low-dimensional representation.  Important differences include that they only study one type of agent whereas we allow agents to have different starting feature vectors, and they assume transferable utility whereas in our setting payments are implicit and do not reduce the utility of the principal.

%% file: prelim.tex
\section{Preliminaries}
\label{sec:prelim}

\paragraph{The Model.}
We denote the true features of an individual by $\vec x\in \R^n$, where the feature space $\R^n$ encodes all relevant features of a candidate, e.g., health history, biometrics, vaccination
record, exercise and dietary habits.
We denote by $f:\R^n \rightarrow [0,1]$  the mapping from one's true features to their true quality. 
For example, a real-valued  $f(\vec x)\in [0,1]$ can express the overall quality of candidates and a binary-valued $f(\vec x)$ can determine whether $\vec x$ meets a certain qualification level.

These true features of an individual may not be visible to the designer.
Instead there is an $n\times n$ projection matrix $P$ of rank $k$, i.e., $P^2 = P$, such that for any $\vec x$, $P\vec x$ represents the visible representation of the individual, such as vaccination record and health
history, but not exercise and dietary habits.
We define by $\mathrm{Img}(P) = \{ \vec z\in \R^n \mid \vec z = P \vec z\}$ the set of all feature representations that are visible to the designer.
We denote by $g: \R^n \rightarrow \R$ a mechanism whose outcome for any individual $\vec x$ depends only on $P\vec x$, i.e., the visible features of $\vec x$.\footnote{
Equiv. $g(\vec x) = g^{||}(P\vec x)$ for unrestricted $g^{||}\!:\!\R^n\!\rightarrow\!\R$.
}
E.g., $g(\vec x) \in \R$ can express the payoff an individual receives from the mechanism, 
$g(\vec x)\in [0,1]$ can express the probability that $\vec x$ is accepted by a randomized mechanism $g$, or  a binary-valued 
$g(\vec x)\in \{0,1\}$ can express whether $\vec x$ is accepted or rejected  deterministically.

Let $\cost(\vec x, \vec x')$ represent the cost that an individual incurs when changing their features from $\vec x$ to $\vec x'$. We consider  $\cost(\vec x, \vec x') = c \|\vec x - \vec x'\|_2$ for some $c>0$. Given mechanism 
$g$, the total payoff $\vec x$ receives from changing its feature representation to $\vec x'$ is
$U_g(\vec x, \vec x') =  g(\vec x') - \cost(\vec x, \vec x').
$
Let $\delta_g:\R^n \rightarrow \R^n$ denote the best response of an individual to  $g$, i.e., 
$ \delta_g(\vec x) = \argmax_{\vec x'}  U_g(\vec x, \vec x').$
We consider a distribution $\D$ over feature vector in $\R^n$, representing individuals.
Our  goal is to design a mechanism $g$ such that, when individuals from  $\D$ best respond to it,
it yields the highest quality individuals on average. That is to find a mechanism $g\in\G$ that maximizes
\begin{equation}
\Val(g) =  \E_{\vec x \sim \D}\big[ f(\delta_g(\vec x)) \big].
\label{eq:exp-quality} 
\end{equation}
We often consider the gain in the quality of individuals compared to the average quality before deploying any mechanism, i.e., the baseline $\E[f(\vec x)]$, defined by 
\begin{equation}
\gain(g) = \Val(g) - \E_{\vec x \sim \D}[ f(\vec x)].
\label{eq:exp-gain} 
\end{equation}
We denote the mechanism with highest gain by  $g_\opt\in \G$.

For example, $f$ can indicate the overall health of an individual and $\G$  the set of governmental policies on how to set insurance premiums based on the observable features of individuals, e.g., setting lower premiums for those who received preventative care in the previous year. Then, $g_\opt$ corresponds to the policy that leads to the healthiest  society on average.
In Sections~\ref{sec:linear-g} and \ref{sec:threshold-g}, we work with different design choices for $f$  and $\G$, and show how to find (approximately) optimal mechanisms.

\paragraph{Types of Mechanisms.}

More specifically, we consider a known true quality function $f$ that is a \emph{linear function} $\vec w_f\cdot \vec x - b_f$ for some vector $\vec w_f \in \R^n$ and $b_f\in \R$. Without loss of generality we assume that the domain and cost multiplier $c$ are scaled so that $\vec w_f$ is a unit vector.

In Section~\ref{sec:linear-g}, we consider the class of \emph{linear mechanisms} $\Glin$, i.e., any $g\in \Glin$ is represented by $g(\vec x ) = \vec w_g \cdot \vec x - b_g$, for a scalar $b_g\in \R$ and  vector $\vec w_g\in \mathrm{Img}(P)$ of length $\| \vec w_g \|_2 \leq R$ for some $R\in \R^+$.
In Section~\ref{sec:threshold-g}, we consider the class of \emph{linear threshold (halfspace) mechanisms} $\Gth$, where a $g\in \Gth$ is represented by $g(\vec x) = \sign(\vec w_g \cdot \vec x - b_g)$  for some unit vector $\vec w_g\in \mathrm{Img}(P)$ and  scalar $b_g\in \R$.

\paragraph{Other Notation.}
$\|\cdot \|$ is the $L_2$ norm of a vector unless otherwise stated.
For $\ell \geq 0$, the \emph{$\ell$-margin density} of a halfspace  is
$ \mden^\ell_\D(\vec w, b) = \Pr_{\vec x \sim \D} \Big[ \vec w \cdot \vec x - b \in [-\ell, 0] \Big].
$
This is the total density $\D$ assigns to points that are at distance at most $\ell$ from being included in the positive side of the halfspace. 
The \emph{soft $\ell$-margin density} of this halfspace is defined as
\[ \sden^\ell_\D(\vec w, b) = 
  \E_{\vec x \sim \D} \Big[ (b - \vec w \cdot \vec x) \ind\big( \vec w \cdot \vec x - b \in [-\ell, 0] \big) \Big].
\]
We suppress $\ell$ and $\D$ when the context is clear.

We assume that individuals have features that are within a ball of radius $r$ of the origin, i.e., $\D$ is only supported on $\X = \{\vec x \mid \| \vec x\| \leq r \}$.
In Section~\ref{sec:threshold-g}, we work with distribution $\D$ that is additionally \emph{$\sigma$-smooth}.
$\D$ is $\sigma$-smoothed distribution if there is a corresponding distribution $\P$
over $\X$ such that to sample $\vec x'\sim \D$ one first samples $\vec x\sim \P$ and then $\vec x' = \vec x + N(0, \sigma^2 I)$. 
Smoothing is a common assumption in theory of computer science where 
$N(0, \sigma^2 I)$ models uncertainties in real-life measurements. 
To ensure that the noise in measurements is small compared to radius of the domain $r$, we assume that $\sigma \in O(r)$.

%% file: linear.tex
\section{Linear Mechanisms}
\label{sec:linear-g}

In this section, we show how the optimal linear mechanism in $\Glin$ is characterized by $g_\opt(\vec x) = \vec w_g \cdot \vec x$ for $\vec w_g$ that is (the largest vector) in the direction of $P \vec w_f$. 
This leads to an algorithm with $O(n)$ runtime for finding the optimal mechanism. 
At a high level, this result shows that when the true quality function $f$ is linear,  the optimal linear mechanism is in the direction of the closest vector to $\vec w_f$ in the visible feature space.  Indeed, this characterization extends to any true quality function that is a monotonic  transformation of a linear function.

\begin{theorem}[Linear Mechanisms]
\label{thm:linear}
Let  $f(\vec x) = h(\vec w_f \cdot \vec x - b_f)$ for some monotonic function $h: \R\rightarrow \R$.
Let $g(\vec x) = \frac{(P\vec w_f)R}{\|P\vec w_f\|_2}  \cdot \vec x$. Then, 
$g$ has the optimal gain.

\end{theorem}

\comment{The proof of Theorem~\ref{thm:linear} follows from basic linear algebra and is included in Appendix~\ref{app:proof_linear} for completeness.
 But let us provide an intuitive explanation of this proof.
In response to mechanism $g(\vec x) = \vec w_g \cdot \vec x - b_g$, any movement of $\vec x$ that is orthogonal to $\vec w_g$ leaves the outcome of the mechanism unchanged but incurs a movement cost to the individual. Therefore, the best-response of any individual to $g$ is to move only in the direction of $\vec w_g$. Each unit of movement in the direction of $\vec w_g$ translates to $\vec w_g \cdot \vec w_f$ units of movement in the direction of $\vec w_f$ and increases the mechanism designer's utility by $\vec w_g \cdot \vec w_f$. So, the mechanism's payoff is maximized by the longest vector in the direction of $\vec w_f$ in the visible feature space.
\footnote{For general norms defined by a convex symmetric set $\mathcal{K}$, 
 $\vec w_g \in \mathcal{K}$ that maximizes $\vec w_g \cdot \vec w_f$ corresponds to the dual of $\vec w_f$ with respect to $\|\cdot \|_{\mathcal{K}}$}
} It is interesting to note that designing an optimal \emph{linear} mechanism (Theorem~\ref{thm:linear}) does not use any information about the distribution of instances in $\D$, rather directly projects $\vec w_f$ on the subspace of visible features.
We see in Section~\ref{sec:threshold-g} that this is not the case for \emph{linear threshold} mechanisms, where the distribution of feature attributes plays a central role in choosing the optimal mechanism.

%% file: threshold.tex
\section{Linear Threshold Mechanisms} \label{sec:threshold-g}

In this section, we consider the class of linear threshold mechanisms $\Gth$ and explore the computational aspects of finding $g_\opt \in \Gth$.
Note that any $g(\vec x) = \sign (\vec w_g \cdot \vec x - b_g)$ corresponds to the transformation of a linear mechanism from $\Glin$ that  only rewards those whose quality passes a threshold. As in the case of linear mechanisms, individuals only move in the direction of $\vec w_g$ and every unit of their movement improves the payoff of the mechanism by $\vec w_g \cdot \vec w_f$. 
However, an individual's payoff from the mechanism improves only if its movement pushes the feature representation from the $0$-side of a linear threshold to $+1$-side.
Therefore only those individuals whose feature representations are close to the decision boundary move.

\begin{lemma} \label{lem:def_delta}
For $g(\vec x) = \sign(\vec w_g \cdot \vec x - b_g)\in \Gth$, 
\begin{align*}
\!\delta_g(\vec x)\!=\! \vec x - \ind\!\left(\! \vec w_g\cdot \vec x\! -\! b_g\! \in\! \left[-\frac 1c, 0\right] \right)\!(\vec w_g\cdot \vec x - b_g) \vec w_g.
\end{align*}
\end{lemma}
This leads to very different dynamics and challenges compared to the linear mechanism case.
For example, $\vec w_g \propto P \vec w_f$ no longer leads to a good mechanism as it may only incentivize a small fraction of individuals as shown in the following example. 

\begin{example}
As in Figure~\ref{fig:bad-ex}, consider $\vec w_f =\left( \frac{1}{\sqrt{3}}, \frac{1}{\sqrt{3}}, \frac{1}{\sqrt{3}}  \right)$ and $P$ that projects a vector on its second and third coordinates. Consider a distribution $\D$ that consists of $\vec x = (0, 0, x_3)$ for $x_3 \sim Unif[-r, r]$. 
Any $g$ only incentivizes individuals who are at distance at most $\ell = 1/c$ from the the decision boundary, highlighted by the shaded regions.
Mechanism $g(\vec x)  = \sign( x_2 - \ell)$ incentivizes  everyone to move $\ell$ unit in the direction of $x_2$ and leads to total utility of $\E[f(\delta_g(\vec x))] = \ell/\sqrt{3}$. 
But, $g'(\vec x) = \sign(\vec w'_g \cdot \vec x - b'_g)$ for unit vector $\vec w_g' \propto P \vec w_f$  only incentivizes $\sqrt{2}\ell/2r$ fraction of the individuals, on average each moves only $\ell/2 \cdot \sqrt{2/3}$ units in the direction of $\vec w_f$. Therefore, $\gain(g') \leq \frac{\ell^2}{2r\sqrt{3}} \ll \gain(g)$ when unit cost $c \gg \frac 1r$. 
\end{example}

\begin{figure}
\centering
\includegraphics[scale=0.4]{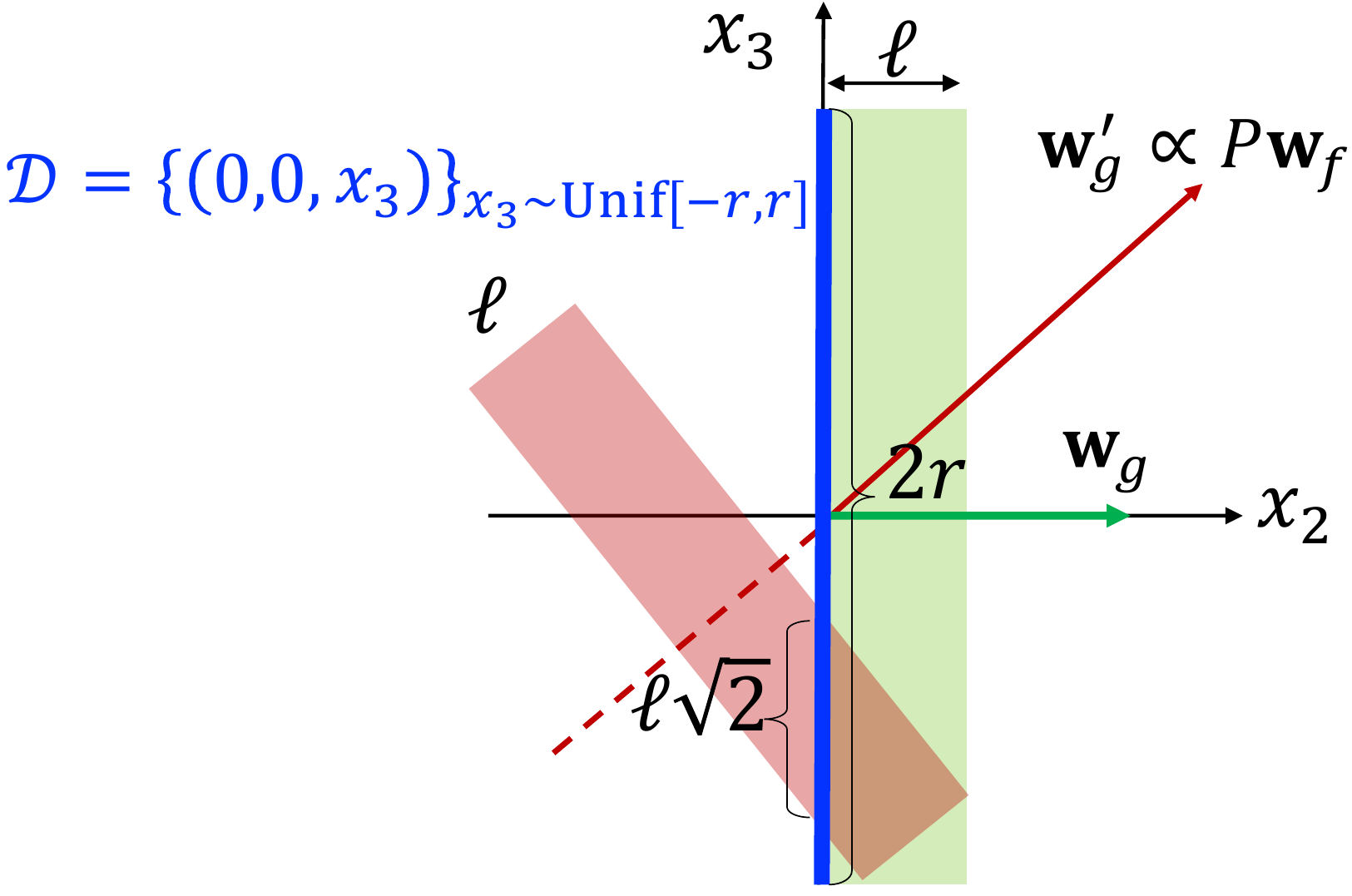}
\caption{{\small Mechanism $g'(\vec x) = \sign( \vec w_g' \cdot \vec x - b_g)$ for any $\vec w'_g \propto P \vec w_f$ is far from optimal.}}
\label{fig:bad-ex}
\end{figure}
Using the characterization of an individual's best-response to $g$ from Lemma~\ref{lem:def_delta},
for any true quality function $f(\vec x) = \vec w_f\cdot \vec x - b_f$ and $g(\vec x) =\sign(\vec w_g \cdot \vec x - b_g)$, we have 
\begin{align}
\gain(g) 
&= (\vec w_g \cdot \vec w_f) \cdot 
\sden^{1/c}(\vec w_g, b_g).
\label{eq:formulation}
\end{align}

While mechanisms with $\vec w_g \propto P\vec w_f$ can be far from the optimal, one can still achieve a $1/(4rc)$-approximation to the optimal mechanism by using such mechanisms and optimizing over the bias term $b_g$.
\begin{theorem}[$1/(4rc)$-approximation]
\label{thm:threshold:large}
Consider the polynomial time algorithm that returns the best $g$ from
$\G = \{ \sign(\frac{P\vec w_f}{\| P \vec w_f\|} \cdot \vec x - b_g) | b_g  = i /2c, \forall i\in \left[ \lceil 2rc \rceil +1 \right]  \}$. Then, $\gain(g) \geq \frac{1}{4rc} \gain (g_\opt)$.
\end{theorem}
This is a good approximation
when the  cost unit $c$ is not too large compared to  the radius of the domain $r$.
However, in most cases the cost to change ones' feature is much larger than a constant fraction of the radius.  In this case to find $g_\opt\in \Gth$, we need to simultaneously optimize the total density of instances that fall within the margin of the classifier, their average distance to the decision boundary, and the correlation between $\vec w_g$ and $\vec w_f$.

One of the challenges involved here is finding a halfspace whose margin captures a large fraction of instances, i.e., has large margin density.
Many variants of this problem have been studied in the past and are known to be hard. For example, the densest subspace, the densest halfspace, densest cube, and the densest ball are all known to be hard to approximate~\citep{ben2002computational,hardt2013algorithms,Johnson78}.
Yet, finding dense regions is a routine unsupervised learning task for which  existing optimization tools  are known to perform well in practice.
Therefore, we assume that we have access to such an optimization tool, which we call \emph{a density optimization oracle.} 
\begin{definition}[Density Optimization Oracle]
Oracle $\mathcal{O}$ takes any distribution $\D$, margin $\ell$, a set $\K \subseteq \R^n$, takes $O(1)$ time and returns
\begin{equation}
 \mathcal{O}(\D, \ell, \K) \in \argmax_{\substack{\vec w\in \K, \|\vec w\| = 1\\b\in \R}} \mden^\ell_\D(\vec w, b).
\end{equation}
\end{definition}

Another challenge  is that $\gain(\cdot)$ is a non-smooth function. As a result, there are distributions for which small changes to $(\vec w, b)$, e.g., to improve $\vec w\cdot \vec w_f$, could result in a completely different gain.
However, one of the properties of real-life distributions is that there is some amount of noise in the data, e.g., because measurements are not perfect. This is modeled by \emph{smoothed distributions}, as described in Section~\ref{sec:prelim}.
Smooth distributions provide an implicit regularization that smooths the expected loss function over the space of all solutions.

Using these two assumptions, i.e.,  access to a density optimization oracle and smoothness of the distribution, we show that there is a polynomial time algorithm that achieves a $1/4$ approximation to the gain of $g_\opt$.
At a high level,  smoothness of $\D$ allows us to limit our search to those $\vec w_g \in \R^n$ for which $\vec w_f \cdot \vec w_g = v$ for a small set of discrete values $v$.
For each $v$, we use the density oracle to search over all $\vec w_g$ such that $\vec  w_f \cdot \vec w_g = v$ and return a candidate with high margin density.
We show how a mechanism with high margin density will lead us to (a potentially different) mechanism with  high soft-margin density and as a result a near optimal gain.
This approach is illustrated in more detail in Algorithm~\ref{alg:opt-w-oracle}.

\newcommand{\C}{\mathcal{C}}
\begin{algorithm}[t]
  \begin{algorithmic}
  \STATE Input:
                $\sigma$-smoothed distribution $\D$ with radius $r$ before perturbation, 
                Vector $\vec w_f$,
                Projection matrix $P$,
                Desired accuracy $\epsilon$.
   \STATE Output: a linear threshold mechanism $g$           
    \STATE{Let $\epsilon' = \min\{\epsilon^4, \epsilon^2 \sigma^4/r^4 \}$ and $\C = \emptyset$.}
    \FOR{$\eta=0,\epsilon' \| P \vec w_f\|, 2\epsilon' \| P \vec w_f\|,\ldots, \|P \vec w_f\|$}
    \STATE{Let $\K_\eta = \mathrm{Img}(P) \cap \{\vec w\mid \vec w\cdot \vec w_f = \eta\}$.}	
    \STATE{Use the density optimization oracle to compute
    \[(\vec w^\eta, b^\eta) \gets \mathcal{O}(\D, \frac{1}{c} ,\K_\eta)
    \]
     }
    \STATE Let $\C \gets \C\cup \{\left(\vec w^\eta, b^\eta), (\vec w^\eta, b^\eta + \frac{1}{2c}\right) \}$
    \ENDFOR
  \RETURN{$g(\vec x) = \sign\left(\vec w\cdot \vec x - b \right)$, where 
    \[ (\vec w, b) \gets \argmax_{(\vec w, b)\in \C} (\vec w\cdot \vec w_f) \sden^{1/c}_\D(\vec w, b).
	\]  
	}
\end{algorithmic}
\caption{$(1/4 - \epsilon)$ Approximation for $\Gth$}
  \label{alg:opt-w-oracle}
\end{algorithm}

\begin{theorem}[\textbf{(Main)} $1/4$-approximation] \label{thm:main-g-threshold}
Consider a distribution over $\X$ and the corresponding $\sigma$-smoothed distribution $\D$ for some $\sigma \in O(r)$.   
Then, for any small enough $\epsilon$, Algorithm~\ref{alg:opt-w-oracle} runs in time 
$\poly(d,  1/\epsilon)$, makes $O(\frac{1}{\epsilon^4}+ \frac{r^2}{\epsilon^2 \sigma^2})$ oracle calls and returns $g \in \Gth$,
such that
\[  \gain(g) \geq \frac 14 \gain(g_\opt) - \epsilon.
\] 
\end{theorem}

\begin{proof}[Proof Sketch of Theorem~\ref{thm:main-g-threshold}]
Recall from Equation~\ref{eq:formulation}, $\gain (g_\opt)$ is the product of two values
$(\vec w_{g_\opt} \cdot \vec w_f)$ and $\sden(\vec w_{g_\opt}, b_{g_\opt})$.
The first lemma shows that we can do a grid search over the value of 
$(\vec w_{g_\opt} \cdot \vec w_f)$.
In other words, there is a predefined grid on the values of $\vec w\cdot \vec w_f$ for which there is a $\vec w$ for which $\vec w\cdot \vec w_f \approx \vec w_{g_\opt}\cdot \vec w_f$. This is demonstrated in Figure~\ref{fig:discrete}.

\begin{lemma}[Discretization] \label{lem:threshold:discretize}
For any two unit vectors $\vec w_1, \vec w_2 \in \mathrm{Img}(P)$, and any $\epsilon$, such that $\vec w_1 \cdot \vec w_2 \leq 1- 2\epsilon$, there is a unit vector $\vec w \in \mathrm{Img}(P)$, such that $\vec w\cdot \vec w_2 = \vec w \cdot \vec w_1 + \epsilon$ and $\vec w \cdot \vec w_1 \geq 1- \epsilon$. 
\end{lemma}

The second technical lemma shows that approximating $\vec w_g$ by a close vector, $\vec w$, and $b_g$ by a close scalar, $b$, only has a small effect on its soft margin density.
That is, the soft margin density is \emph{Lipschitz}.
For this claim to hold, it is essential for the distribution to be smooth.
The key property of a $\sigma$-smoothed distribution $\D$---corresponding to the original distribution $\P$---is that $\mden_\D^\ell(\vec w_g, b_g)$ includes instances $\vec x\sim \P$ that are not in the margin of $(\vec w_g, b_g)$.
As shown in Figure~\ref{fig:discrete}, these instances also contribute to the soft margin density of any other $\vec w$ and $ b$, as long as the distance of $\vec x$ to halfplanes $\vec w \cdot  \vec x = b$ and  $\vec w_g \cdot \vec x = b_g$ are comparable.
So, it is sufficient to show that the distance of any $\vec x$ to two halfplanes with a small angle is approximately the same. Here, angle between two unit vectors is defined as $\theta(\vec w, \vec w') = \arccos(\vec w \cdot \vec w')$. 
Lemma~\ref{lem:threshold:discretize} leverages this fact to prove that soft margin density is Lipschitz smooth. 

\begin{lemma}[Soft Margin Lipschitzness] \label{lem:threshold:soft-margin-lipschitz}
For any distribution over $\X$ and its corresponding $\sigma$-smooth distribution $\D$, for $\epsilon \leq \frac{1}{3R^2}$,  any $\nu< 1$,  for  $R = 2r + \sigma \sqrt{2\ln(2/\nu)}$, and any $\ell \leq O(R)$, 
if $\theta(\vec w_1, \vec w_2) \leq \epsilon \sigma^2$ and  $|b_1 - b_2| \leq \epsilon \sigma$,  we have  
\[  \left| \sden_{\D}^\ell(\vec w_1, b_1)  - \sden_{\D}^\ell(\vec w_2, b_2)    \right| 
  \leq O(\nu + \epsilon R^2).
  \]
\end{lemma}

\begin{figure}
\centering
\includegraphics[scale=0.8]{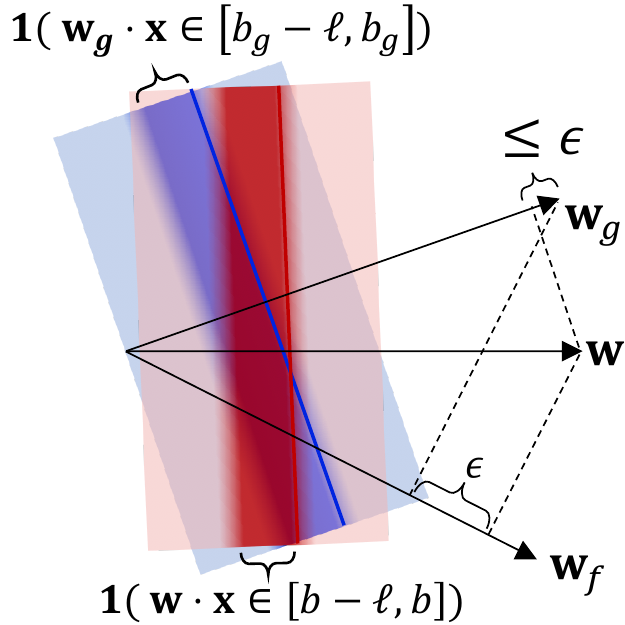}
\caption{{\small The $\vec w_g \cdot \vec w_f$ can be increased by any small amount $\epsilon$ to match a specific value,
when $\vec w_f$ is shifted by a $\leq \epsilon$ angle to vector $\vec w$.
When a distribution is smooth, the margin density of mechanisms $(\vec w, b)$ and $(\vec w_g, b_g)$ are close when
$\theta(\vec w, \vec w_g)$ and $|b- b_g|$ are small.
The lightly shaded area shows that instances outside of the margin contribute to (soft-)margin density.
}}
\label{fig:discrete}
\end{figure}

The final step to bring these results together is 
to show that we can find a mechanism (between those on the grid) that has a large soft margin density.
To do this, we show that the mechanism (with potentially small changes to it) that has the highest
margin density has at least $1/4$ of the soft density of the optimal mechanism.
The proof technique of this lemma involves analyzing whether more than half of instances within an $\ell$-margin of a mechanism 
are at distance at most $\ell/2$ of the margin, in which case soft margin density of the mechanism is at least $1/4$ of its density. Otherwise, shifting the bias of the mechanism by $\ell/2$ results in a mechanism whose soft margin is a least $1/4$ of the original margin density.

\begin{lemma}[Margin Density Approximates Soft-Margin Density]
\label{lem:threshold:1/4-approx}
For any distribution $\D$ over $\R^n$, a class of unit vectors $\K\subseteq \R^n$, margin $\ell > 0$, and a set of values $N$,  let $(\vec w_\eta,  b_\eta) = \mathcal{O}(\D, \ell, \K \cap\{ \vec w\cdot \vec w_f = \eta\})$ for $\eta\in N$, and let $(\hat{\vec w}, \hat{b})$ be the mechanism with maximum $\eta \mden^\ell_\D(\vec w_\eta,  b_\eta)$ among these options.
Then, 
\[\!\!\max\!\left\{\!\!\gain(\hat{\vec w}, \hat b),\!\gain(\hat{\vec w}, \hat b +\frac{\ell}{2})\!\! \right\}\!\!
\geq \!\! \frac 14\!\!
\max_{\substack{\vec w\in \K\\ \vec w\cdot \vec w_f \in N \\b\in \R}}\!\!\!\gain(\vec w, b).
\]
\end{lemma}
Putting these lemmas together, one can show that Algorithm~\ref{alg:opt-w-oracle}, finds a $\frac 14$ approximation.
\end{proof}

%% file: learning.tex
\section{Learning Optimal Mechanisms From Samples}
\label{sec:learn}
Up to this point, we have assumed that distribution $\D$ is known to the mechanism designer and all computations, such as measuring margin density and soft margin density, can be directly performed on $\D$.
However, in most applications the underlying distribution over feature attributes 
is unknown or the designer only has access to a small set of historic data.
In this section, we show that all computations that are performed in this paper can be done instead on a small set of samples that are drawn i.i.d. from distribution $\D$. Note that (since no mechanism is yet deployed at this point) these samples represent the initial non-strategic feature attributes.

We first note that the characterization of the optimal \emph{linear} mechanism (Theorem~\ref{thm:linear}) is independent of distribution $\D$, that is, even with no samples from $\D$ we can still design the optimal linear mechanism.
On the other hand, the optimal \emph{linear threshold} mechanism (Theorem~\ref{thm:main-g-threshold}) heavily depends on the distribution of instances, e.g., through the margin and soft-margin density of $\D$.
To address sample complexity of learning a linear threshold mechanism, we use the concept of \emph{pseudo-dimension}. This is the analog of VC-dimension for real-valued functions.

\begin{definition}[Pseudo-dimension]
Consider a set of real-valued functions $\F$ on the instance space $\X$.  
A set of instances $x^{(1)}, \dots, x^{(n)}\in \X$ is \emph{shattered} by $\F$,
if there are values $v^{(1)}, \dots, v^{(n)}$ such for any subset $T\subseteq[n]$ there is a function $f_T\in \F$ for which $f(x^{(i)}) \geq v^{(i)}$ if and only if $i\in T$.
The size of the largest set that is shattered by $\F$ is the \emph{pseudo-dimension} of $\F$.
\end{definition}

It is well-known that the pseudo-dimension of a function class closely (via upper and lower bounds) controls the number of samples that are needed for learning a good function. This is characterized by the uniform-convergence property as shown below.

\begin{theorem}[Uniform Convergence~\cite{Pol84}]
\label{thm:pollard}
Consider the class of functions $\F$ such that $f: \X \rightarrow [0, H]$ with pseudo-dimension $d$. For any distribution $\D$ and a set $S$ of $m = \epsilon^{-2} H^2(d + \ln(1/\epsilon))$ i.i.d. randomly drawn samples from $\D$, with probability $1-\delta$, for all $f\in \F$
\[  \left| \frac{1}{m} \sum_{x\in S} f(x) - \E_{x\sim \D}[f(x)] \right|  \leq \epsilon.
\]
\end{theorem}

As is commonly used in machine learning, uniform convergence implies that choosing any mechanism based on its performance on the sample set leads to a mechanism whose expected performance is within $2\epsilon$ of the optimal mechanism for the underlying distribution.

\begin{lemma}[Pseudo-dimension]
\label{lem:pseudo-mech}
For each $g= sign(\vec w_g\cdot \vec x - b_g)$ and $\vec x$ let $\gain_{\vec x}(g)$ 
be the gain of $g$ just on instance $\vec x$ and let $\mden_{\vec x}^\ell(g) = \ind(\vec w_g\cdot \vec x- b_g\in [-\ell, 0])$. 
The class of real-valued functions $\{\gain \circ g\}_{g\in \Gth}$ has a pseudo-dimension that is at most $O(\mathrm{rank}(P))$
 Moreover, the class of functions $\{ \vec x \rightarrow \mden_{\vec x}^\ell(g) \}_{g\in \Gth}$ has a pseudo-dimension (equivalently VC dimension)  $O(\mathrm{rank}(P))$.
\end{lemma}
\begin{figure}
\centering
    \includegraphics[scale = 0.8]{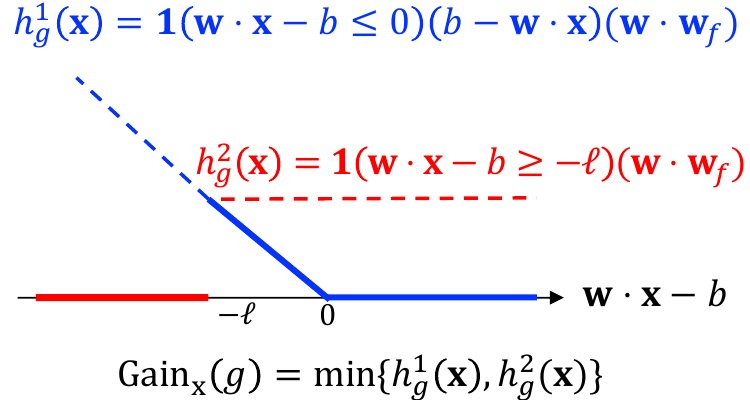}
  \caption{\small $\gain_{\vec x}(g)$ is a minimum of $h^1_g(\vec x)$ and $h^2_g(\vec x)$.}
    \label{fig:pseudo}
\end{figure}

\begin{proof}
For $g\in \Gth$, note that 
\[
\!\gain_{\vec x}(g) \!=\! \ind\!\!\left(\!\vec w_g\! \cdot\! \vec x - b_g \!\in \!\left[-\frac 1c, 0\right] \right)\!(b_g- \vec w_g \cdot \vec x)(\vec w_f \cdot \vec w_g).
\]
Note that we can write $\gain_{\vec x}(g) = \min\{ h_1(\vec x), h_2(\vec x)\}$ for 
$h^1_g(\vec x) := \ind(\vec w \cdot \vec x - b \leq 0)(b - \vec w \cdot \vec x)  (\vec w\cdot \vec w_f)$ and $h^2_g (\vec x) := \ind( \vec w\cdot \vec x - b \geq - \ell) (\vec w\cdot \vec w_f)$.
As show in Figure~\ref{fig:pseudo}, $h^1_g$ and $h^2_g$ are both monotone functions of $\vec w\cdot \vec x - b$.

Pollard~\citep{Pol84} shows that the set of all linear functions in a rank $k$ subspace has pseudo-dimension $k$. Pollard~\citep{Pol84}
also shows that the set of monotone transformations of these  functions  has pseudo-dimension $O(k)$. Therefore, the set of functions $\{h^1_g(\vec x)\}_{g\in \Gth} $ and  $\{h^2_g(\vec x)\}_{g\in \Gth} $, each have pseudo-dimension $\mathrm{Rank}(P)$. 
It is well-known that the set of all minimums of two functions from two classes each with pseudo-dimension $d$ has a pseudo-dimension of $O(d)$. Therefore, $\{ \vec x \rightarrow \gain_{\vec x}(g)\}_{g\in \Gth}$ has pseudo-dimension of at  most $O(\mathrm{rank}(P))$.
A similar construction shows that $\{ \vec x \rightarrow \mden_{\vec x}^\ell(g) \}_{g\in \Gth}$ has VC dimension of at  most $O(\mathrm{rank}(P))$. 
\end{proof}
We give a bound on the number of samples sufficient for finding a near optimal mechanism in $\Gth$.
\begin{theorem}[Sample Complexity]
\label{thm:sample}
For any small enough $\epsilon$ and $\delta$, 
there is 
\[
m \in O\left(\frac{1}{c^2 \epsilon^2}\left(\mathrm{rank}(P) + \ln(1/\delta) \right)\right)
\] 
such that for $S\sim \D^m$ i.i.d. samples with probability $1-\delta$, 
$\hat g\in \Gth$ that optimizes the empirical gain on $S$ has $\gain (\hat g) \geq \gain(g_\opt) - \epsilon$.
Furthermore, with probability $1-\delta$ when  Algorithm~\ref{alg:opt-w-oracle} is run on $S$, the outcome $\hat g$ has $\gain (\hat g) \geq \frac 14 \gain(g_\opt) - \epsilon$.
\end{theorem}

%% file: discussion.tex
\section{Discussion}

In this work, we focused on increasing the expected quality of the population, i.e., welfare, through classification.
There are, however, other reasonable objectives that one may want to consider.
In some cases, our goal may be to design 
evaluation mechanisms that lead to highest quality \emph{accepted} individuals, not accounting for those who are rejected by the mechanism. 
It would be interesting to study these objectives under our model as well. An example of this is to consider, for a set of boolean valued functions $\G$,
$
\max_{g\in \G} \E_{\vec x \sim \D}\big[ f(\delta_g(\vec x))  \mid g(\vec x) = 1 \big].
$ 

Throughout this work, we focused on the $L_2$ cost function (i.e., $\cost(\vec x, \vec x') = c \|\vec x - \vec x'\|_2$).
This specification models scenarios in which an individual can improve multiple features most efficiently by combining effort
rather than working on each feature individually ($L_1$ norm).  
For example, simultaneously improving  writing, vocabulary, and critical analysis of a student more by for example reading novels, is more effective than spending effort to improve 
vocabulary by, say, memorizing vocab lists, and then improving the other attributes.
It would be interesting to analyze the algorithmic aspects of alternative cost functions, such as $L_1$ cost or even
non-metric costs 
(e.g., ones with learning curves whereby the first $10\%$ improvement is cheaper than the next $10\%$), 
and different costs for different types of individuals.

Finally, we have assumed we know the true mapping of features to qualities (i.e., $f$).  In many settings, one might not know this mapping, or even the full set of features.  Instead, the designer only observes the quality of individuals after they respond to their incentives (i.e., $(f(\delta_g(\vec x)))$), and the projection of their new feature set (i.e., $P \delta_g(\vec x)$).  
Existing works on Stackelberg games and strategic classification have considered the use of learning tools for designing optimal mechanisms without the knowledge of the underlying function~\citep{BHP14,HaghtalabFNSPT16,miller2019strategic}.
It would be interesting to characterize how the nature of observations and interventions available in our setting specifically affects this learning process.

%% file: app_smooth.tex
\section{Useful Properties of Smoothed Distributions}
\label{sec:app:smooth}

\begin{lemma} \label{lem:smooth-band}
For any $\sigma$-smoothed distribution $\D$ and any unit vector $\vec w$ and any range $[a, b]$,
\[ \Pr_{\vec x\sim \D}[\vec w\cdot \vec x \in [a,b] ]\leq \frac{|b-a|}{\sigma\sqrt{2\pi}}.
\]
\end{lemma}
\begin{proof}
First note that for any Gaussian distribution with variance $\sigma^2 I$ over $\vec x$, the distribution of $\vec w\cdot \vec x$ is a $1$-dimensional Gaussian with variance $\sigma^2$. Since any $\sigma$-smoothed distribution is a mixture of many Gaussians, distribution over $\vec w \cdot \vec x$ is also the mixture of many Gaussians. It is not hard to see that the density of any one dimensional Gaussian is maximized in range $[a, b]$ if its centered at $(a-b)/2$. Therefore,
\begin{align*}
\Pr_{\vec x\sim \D}[\vec w\cdot \vec x \in [a,b]] 
 & \leq  \Pr_{x\sim N(0, \sigma^2)} \left[ \frac{a-b}{2} \leq x \leq \frac{b-a}{2} \right] \\
 & \leq  \frac{2}{\sigma \sqrt{2\pi}} \int_0^{(b-a)/2} \exp\left(-\frac{z^2}{2\sigma^2}\right) ~dz\\
 & \leq  \frac{(b-a)}{\sigma \sqrt{2\pi}}.
\end{align*} 
\end{proof}

\begin{lemma} \label{lem:smooth-far}
For any distribution over $\{ \vec x\mid \| \vec x\| \leq r\}$ and its corresponding $\sigma$-smoothed distribution $\D$, and $H \geq 2r + \sigma \sqrt{2 \ln (1/\epsilon)}$, we have
\[ \Pr_{\vec x\sim \D}\big[ \| \vec x\| \geq H \big] \leq \epsilon. 
\]
\end{lemma}
\begin{proof}
Consider the Gaussian distribution $N( \vec c, \sigma^2I)$ for some $\| \vec c\| \leq r$. Then for any $\|\vec x\| \geq H$, it must be that the distance of $\vec x$ to $\vec c$ is at least $H - 2r$. Therefore,
\begin{align*}
\Pr_{\vec x\sim \D}[\|\vec x\| \geq H] \leq 
\Pr_{x\sim N(0, \sigma^2)}[x \geq H - 2r] \leq \exp\left( - \frac{(H- 2r)^2}{2\sigma^2}  \right) \leq \epsilon.
\end{align*}
\end{proof}

\begin{lemma} \label{lem:smooth-den/den}
For any distribution over $\{ \vec x\mid \| \vec x\| \leq r\}$ and let $p(\cdot)$ denote the density of the corresponding $\sigma$-smoothed distribution. Let $\vec x$ and $\vec x'$ be such that $\| \vec x - \vec x'\| \leq D$. Then, 
\[ \frac{p(\vec x)}{p( \vec x')} \leq  \exp\left( \frac{D (2 r + \|\vec x\| + \|\vec x'\|)}{2 \sigma^2} \right).
\]
\end{lemma}
\begin{proof}
Consider distribution $N(\vec c, \sigma^2 I)$ for some $\|\vec c\| \leq r$.
Without loss of generality assume that $\| \vec c - \vec x\| \leq \| \vec c - \vec x'\|$, else $ p(\vec x)/p( \vec x') \leq 1$. 
We have
\begin{align*}
\frac{p(\vec x)}{p(\vec x')}
	& = \frac{\exp\left( \frac{-\| \vec c - \vec x\|^2}{2\sigma^2} \right) }{\exp\left( \frac{-\| \vec c - \vec x'\|^2}{2\sigma^2} \right) } \\
	& = \exp\left(\frac{\| \vec c - \vec x'\|^2 - \| \vec c - \vec x\|^2 }{2\sigma^2}    \right) \\
	& = 	\exp\left(\frac{\left(\| \vec c - \vec x'\| - \| \vec c - \vec x\|\right) \left(\| \vec c - \vec x'\| + \| \vec c - \vec x\| \right) }{2\sigma^2}    \right) \\
	& \leq \exp\left( \frac{D (2 r + \|\vec x\| + \|\vec x'\|)}{2 \sigma^2} \right).
\end{align*}

Since $\D$ is a $\sigma$-smoothed distribution, it is a mixture of many Gaussians with variance $\sigma^2 I$ centered within $L_2$ ball of radius $r$. Summing over the density for all these Gaussian proves the claim. 
\end{proof}

%% file: app_proof_linear.tex
\section{Proofs from Section~\ref{sec:linear-g}}
\label{app:proof_linear}

\begin{lemma} \label{lem:g_linear:delta}
For any linear $g(\vec x) = \vec w_g \cdot \vec x - b_g$ and cost function $\cost(\vec x, \vec x') = c \|\vec x - \vec x'\|_2$, 
\begin{align*}
   \delta_g(\vec x) =  \vec x + \alpha \vec w_g,
\end{align*} 
where $\alpha = 0$ if $c > \| \vec w_g \|$ and $\alpha = \infty$ if $c \leq \| \vec w_g \|$. 
\end{lemma}
\begin{proof}

Consider an image of a vector on $\vec w_g$ and its orthogonal space, note that any $\vec x'$ can be represented by $\vec x' - \vec x = \alpha \vec w_g + \vec z,$
where $\vec z\cdot \vec w_g = 0$ and $\alpha \in \R$. Then, the payoff that a player with features $\vec x$ gets from playing $\vec x'$ is
\begin{align*}
U_g (\vec x, \vec x') &= \vec w_g \cdot \vec x' - b_g - c\| \vec x - \vec x'\| \\
	&= \vec w_g\cdot \left( \alpha \vec w_g + \vec z \right) - b_g - c\|   \alpha \vec w_g + \vec z \|  \\
	& = \alpha \| \vec w_g \|^2 - b_g - \alpha c \|\vec w_g\| - |\alpha|c \|\vec z\|, 
\end{align*}
where the last transition is due to the fact that $\vec z$ is orthogonal to  $\vec w_g$.

It is clear from the above equation that the player's utility is optimized only if $\vec z = \vec 0$ and $\alpha \geq 0$, that is, $\delta_g(\vec x)  = \vec x + \alpha \vec w_g$ for some $\alpha \geq 0$. Therefore,
\[ \delta_g(\vec x) = \argmax_{\vec x'}~ \vec w_g \cdot \vec x - b_g - c\|\vec x - \vec x'\| = \vec x + \left(\argmax_{\alpha\geq 0} ~ \alpha \|\vec w_g\| \left( \| \vec w_g \| - c \right) \right) \vec w_g.
\]
Note that $\alpha \|\vec w_g\| \left( \| \vec w_g \| - c \right)$ is maximizes at $\alpha = 0$ if $c > \| \vec w_g \|$ and is maximized at $\alpha = \infty$ if $c \leq \| \vec w_g \|$.
\end{proof}

\subsection{Proof of Theorem~\ref{thm:linear}}

For ease of exposition we refer to parameters of $g_\opt$ by $\vec w^*$ and $b^*$.
Let $\vec w_g = \frac{(P\vec w_f)R}{\|P\vec w_f\|_2}$ and $g(\vec x) = \vec w_g \cdot \vec x$.
By the choice of $\vec w_g$ and definition of $\G$,  we have  $\|\vec w_g\| = R \geq \| \vec w^* \|$ and $\vec w_g \in \mathrm{Img}(P)$.

By Lemma~\ref{lem:g_linear:delta} and the fact that $ \|\vec w_g \| \geq  \| \vec w^* \|$ there are $\alpha_g, \alpha_{g_\opt}\in\{0, \infty\}$,  such that 
$\alpha_{g} \geq \alpha_{g_\opt}$,
$\delta_{g_\opt}(\vec x) = \vec x + \alpha_{g_\opt} \vec w_g$, and $\delta_{g}(\vec x) = \vec x + \alpha_{g}\vec w_g.$
Furthermore, we have 
\begin{align*}
 \vec w_f \cdot \vec w^* \leq  \| \vec w_f\| \|\vec w^*\| \leq  \| \vec w_f\| R = \vec w_f\cdot \frac{R(P \vec w_f)}{\|P \vec w_f\|}  =  \vec w_f \cdot \vec w_g,
\end{align*}
where the first transition is by Cauchy-Schwarz, the second transition is by the definition of $\G$, and the last two transitions are by the definition of $\vec w_g$.
Using the monotonicity of $h$ and the fact that $\alpha_g\geq \alpha_{g_\opt}$ and $\vec w_f \cdot \vec w_g \geq  \vec w_f \cdot \vec w^*$,  we have
\begin{align*}
\E_{\vec x \sim \D} \big[ f(\delta_{g_\opt}(\vec x) \big] 
	&=     \E_{\vec x \sim \D}\big[ h\left( \vec w_f\cdot \vec x + \alpha_g  \vec w_f \cdot \vec w^* - b_f \right) \big] \\
	& \leq \E_{\vec x \sim \D}\big[ h\left( \vec w_f\cdot \vec x + \alpha_{g_\opt}  \vec w_f \cdot \vec w_g - b_f \right) \big]\\
	& = \E_{\vec x \sim \D}\big[ f(\delta_{g}(\vec x)) \big].
\end{align*}
This proves the claim.

%% file: app_proof_threshold.tex
\section{Proofs from Section~\ref{sec:threshold-g}}
\label{app:threshold:proofs}

\subsection{Proof of Lemma~\ref{lem:def_delta}}
By definition, the closest $\vec x'$ to $\vec x$ such that $g(\vec x') = 1$ is the projection of $\vec x$ on $g$, defined by $\vec x'  = \vec x - (\vec w\cdot \vec x - b) \vec w$. By definition 
\[
\cost(\vec x, \vec x') = c\|\vec x - \vec x'\|_2 = c \left(\vec w\cdot \vec x - b\right).
\]
The claim follows by noting that a player with features  $\vec x$ would report $\delta_g(\vec x) = \vec x'$ only if $\cost(\vec x, \vec x')\leq 1$, otherwise $\delta_g(\vec x) = \vec x$.

\subsection{Proof of Lemma~\ref{lem:threshold:discretize}}
\label{app:threshold:proof:discretize}

The proof idea is quite simple: Only a small move in angle is required to achieve an $\epsilon$ change in the projection of on $\vec w_2$. Here, we prove this rigorously.  Let $\vec w = \alpha \vec w_1 + \beta \vec w_2$ for $\alpha$ and $\beta$ that we will described shortly. Note that $\mathrm{Img}(P)$ includes any linear combination of $\vec w_1$ and $\vec w_2$. Therefore, $\vec w\in P$. Let $\omega = \vec w_1 \cdot \vec w_2$, 
\[ \alpha = \sqrt{\frac{1 - (\omega + \epsilon)^2}{1 - \omega^2}},\]
and 
\[ \beta = \sqrt{1 - \alpha^2 + \alpha^2 \omega^2} - \alpha \omega.\]

By definition of $\alpha$ and $\beta$, $\| \vec w\| = \sqrt{\alpha^2 + \beta^2 + 2\alpha \beta \omega}  = 1$.
That, is $\vec w$ is indeed a unit vector.
Next, we show that $\vec w \cdot \vec w_2 = \vec w_1 \cdot \vec w_2 + \epsilon$. We have
\begin{align*}
\vec w \cdot \vec w_2 &= \alpha \omega + \beta = \sqrt{\frac{\omega ^2 \left(1-(\epsilon+\omega)^2\right)}{1-\omega^2}-\frac{1-(\epsilon+\omega)^2}{1-\omega^2}+1}  = \omega + \epsilon = \vec w_1 \cdot \vec w_2 + \epsilon.
\end{align*}

It remains to show that $\vec w_1 \cdot \vec w \geq 1-\epsilon$. By definition of $\vec w$, we have
\begin{align*}
\vec w\cdot \vec w_1 &= \alpha + \beta \omega  =  \left(1-\omega^2\right) \sqrt{\frac{(\epsilon+\omega)^2-1}{\omega^2-1}}+\omega (\epsilon+\omega). \\
\end{align*}
One can show that the right hand-side equation is monotonically decreasing in $\omega$
Since $\omega < 1 - 2\epsilon$, we have that 
\begin{align*}
\vec w\cdot \vec w_1 
	&\geq  \left(1-(1- 2\epsilon)^2\right) \sqrt{\frac{1 -(1- \epsilon)^2}{1- (1- 2\epsilon)^2}}+ (1- 2\epsilon) (1- \epsilon) \\
	& =  (1-\epsilon) \left(2 \epsilon \left(\sqrt{\frac{2-\epsilon}{1 - \epsilon}}-1\right)+1\right)\\
	&\geq 1- \epsilon.
\end{align*}
This completes the proof.

\subsection{Proof of Lemma~\ref{lem:threshold:soft-margin-lipschitz}}
\label{app:threshold:proof:lipschitz}

Without loss of generality, let the angle between $\vec w_1$ and $\vec w_2$ to be exactly $\beta = \epsilon \sigma^2$. Also without loss of generality assume $b_1 > b_2$ and let $\vec w_1 = (1, 0, \dots, 0)$ and $\vec w_2 = (\cos(\beta), \sin(\beta), 0, \dots, 0)$. Consider the following rotation matrix
\[  M = 
\begin{pmatrix}	
	\cos(\beta)  & \sin (\beta) & 0 & \cdots & 0\\
	-\sin(\beta)  & \cos (\beta) & 0 & \cdots & 0 \\
    0             & 0             & 1 & \cdots & 0 \\
    \vdots        & \vdots        &   &\ddots &   \\
    0             & 0             & 0 & \cdots &   1\\
\end{pmatrix}
\]
that rotates every vector by angle $\beta$ on the axis of $\vec w_1, \vec w_2$. Note that $\vec x\rightarrow P\vec x$ is a bijection. Furthermore, for any $\vec x$, $\vec w_1 \cdot \vec x = \vec w_2 \cdot (M\vec x)$ and the $\|\vec x\| = \| M \vec x\|$.

We use the mapping $\vec x\rightarrow P\vec x$ to upper bound the difference between the soft-margin density of $(\vec w_1, b_1)$ and $(\vec w_2, b_2)$. At a high level, we show that the total contribution of $\vec x$ to soft-margin density of $(\vec w_1, b_1)$ is close to the total contribution of $M\vec x$ to the soft-margin density of $(\vec w_2, b_2)$. 

For this to work, we first show that almost all instances $\vec x$ in the $\ell$-margin of $(\vec w_1, b_1)$ translate to instances $M\vec x$ in the $\ell$-margin of $(\vec w_2, b_2)$, and vice versa. That is,
\begin{align*} \label{eq:match-bands1}
\Pr_{\vec x\sim \D}\Big[ \vec w_1\cdot \vec x - b_1 \in [-\ell, 0] \text{ but } \vec w_2\cdot M \vec x - b_2 \notin [-\ell, 0], \|\vec x\|\leq R \Big]  & \leq  \Pr_{\vec x\sim \D}\Big[ \vec w_1\cdot \vec x - b_1 \in [-\sigma \epsilon, 0]  \Big] \\
       & \leq \frac{\epsilon}{\sqrt{2\pi}}, \numberthis
\end{align*}
where the first transition is due to that fact that $\vec w_1\cdot \vec x = \vec w_2 \cdot M\vec x$, so only instances for which  
$\vec w_1\cdot \vec x - b_1 \in [-\sigma \epsilon, 0]$ contribute to the above event. The last transition is by Lemma~\ref{lem:smooth-band}. Similarly, 
\begin{equation}\label{eq:match-bands2}
\Pr_{\vec x\sim \D}\Big[ \vec w_1\cdot \vec x - b_1 \notin [-\ell, 0] \text{ but } \vec w_2\cdot M \vec x - b_2 \in [-\ell, 0], \|\vec x\| \leq R \Big]                        \leq \frac{\epsilon}{\sqrt{2\pi}}.
\end{equation}
Next, we show that the contribution of $\vec x$ to the soft-margin of $(\vec w_1, b_1)$ is close to the contribution of $M\vec x$ to the soft-margin of $(\vec w_2, b_2)$.

\begin{equation} \label{eq:error-b-w.x}
\left| (b_1 - \vec w_1 \cdot  \vec x ) - (b_2 - \vec w_2 \cdot M \vec x) \right| \leq |b_1 - b_2| \leq \sigma \epsilon.
\end{equation}
Moreover, $\vec x$ is also close to $M \vec x$, since $M$ rotates any vector by angle $\beta$ only. That is, for all $\vec x$, such that $\| \vec x\| \leq R$, we have
\begin{align*}
\left\| \vec x - M \vec x \right\| \leq \| \vec x\| \| I - M \|_F  \leq 2\sqrt{2} \sin(\beta/2) R \leq \sqrt{2} \beta R.
\end{align*}

Now consider the density of distribution $\D$ indicated by $p(\cdot)$. By Lemma~\ref{lem:smooth-den/den} and the fact that $\|\vec x\| = \| M \vec x\|$, we have that for any $\| \vec x\| \leq R$, 
\begin{equation} \label{eq:p(x)/p(x')}
 \frac{p(\vec x)}{p(M \vec x)} \leq \exp\left(\frac{\sqrt{2}\beta R \cdot (6r + 2\sigma\sqrt{2\ln(2/\epsilon)})}{2\sigma^2} \right) \leq \exp\left(\frac{3\beta R^2}{\sigma^2}\right) <  1 + 6 \beta \frac{R^2}{\sigma^2} = 1+ 6 \epsilon R^2
\end{equation} 
where the penultimate transition is by the fact that $\exp(x) \leq 1+2x$ for $x< 1$ and $\frac{3\beta R^2}{\sigma^2}\leq 1$ for $\epsilon \leq 1/3R^2$. Lastly, by Lemma~\ref{lem:smooth-far}, all but $\nu$ fraction of the points in $\D$ are within distance $R$ of the origin.
Putting these together, for $\Omega = \left\{\vec x\mid \| \vec x\|\leq R\right\}$, we have
\begin{align*}
\big| & \sden_{\D}(\vec w_1, b_1)  - \sden_{\D}(\vec w_2, b_2) \big| \\
& = 
 \left| \E_{\vec x \sim \D}\Big[ (b_1 - \vec w_1\cdot \vec x) \ind(\vec w_1\cdot \vec x - b_1\in [-\ell, 0])\Big] -  \E_{M\vec x \sim \D}\Big[(b_2 - \vec w_2\cdot M\vec x) \ind(\vec w_2\cdot M\vec x - b_2\in [-\ell, 0])    \Big]    \right| \\
	& \leq  \Pr_{\vec x\sim \D}\left[ \vec x\notin \Omega  \right] + \ell \Pr_{\vec x\sim \D}\Big[ \vec w_1\cdot \vec x - b_1 \in [-\ell, 0] \text{ but } \vec w_2\cdot M \vec x - b_2 \notin [-\ell, 0] \Big] \\
  	&\quad  + \ell \Pr_{M\vec x\sim \D}\Big[ \vec w_1\cdot \vec x - b_1 \notin [-\ell, 0] \text{ but } \vec w_2\cdot M \vec x - b_2 \in [-\ell, 0] \Big] \\
    & \quad + \left| \int_\Omega  \Big(    p(\vec x) (b_1 - \vec w_1\cdot \vec x) - p(M\vec x) (b_2 - \vec w_2\cdot M\vec x)  \Big)~ d\vec x \right| 
    \\
    & \leq \nu + \frac{2 \ell \epsilon}{\sqrt{2\pi}} + 
    2 \int_{\Omega} \max\{ |p(\vec x) - p(M\vec x)|, |(b_1 - \vec w_1\cdot \vec x)  - (b_2 - \vec w_2 \cdot M \vec x)| \} \\
    & \leq  \nu + \frac{2 \epsilon \ell}{\sqrt{2\pi}} + 
   6R^2 \epsilon + \sigma \epsilon\\
   & \leq O(\nu +\epsilon R^2),
\end{align*}
where the last transition is by the assumption that $\sigma \in O(R)$ and $\ell \in O(R)$.

\subsection{Proof of Lemma~\ref{lem:threshold:1/4-approx}}
\label{app:threshold:proof:soft}

We suppress $\D$ and $\ell$ in the notation of margin density and soft-margin density below.
Let $\vec w^*, b^*$ be the optimal solution  to the weighted soft-margin density problem, i.e.,
\[ \left(\vec w^*\cdot \vec w_f\right) \cdot \sden_\D^\ell\left(\vec w^*, b^* \right) = \max_{\eta\in N} \eta \cdot \max_{\substack{\vec w\cdot \vec w_f =\eta \\ b\in \R}} ~ \left(\vec w\cdot \vec w_f \right)\cdot \sden(\vec w, b).  
\]
Note that for any $\vec x\sim \D$ whose contribution to $\sden(\vec w^*, b^*)$ is non-zero, it must be that $\vec w^*\cdot \vec x - b\in [-\ell, 0]$. By definitions of margin and soft margin densities, we have
\begin{equation} \label{eq:proof-soft-margin1}
  \alpha\left( \vec w^*\cdot \vec w_f\right)\cdot  \sden(\vec w^*, b^*) \leq   \alpha \ell \left( \vec w^*\cdot \vec w_f\right) \cdot  \mden(\vec w^*, b^*)  \leq \ell  \left( \hat{\vec w}\cdot \vec w_f\right) \cdot \mden(\hat {\vec w}, \hat b).
\end{equation}
 
Consider the margin density of $(\hat{\vec w}, \hat b)$. There are two cases:
\begin{enumerate}
\item At least half of the points, $\vec x$, in the $\ell$-margin are at distance at least $\ell/2$ from the boundary, i.e., $\hat{\vec w}\cdot \vec x - \hat b \in [-\ell, -\ell/2]$. In this case, we have
\begin{equation}
\label{eq:proof-soft-margin1.2a}
 \left( \frac{1}{2} \right) \left( \frac \ell 2\right) \mden(\hat {\vec w}, \hat b)  \leq \sden(\hat {\vec w}, \hat b).    
\end{equation}
\item At least half of the points, $\vec x$, in the $\ell$-margin are at distance at most $\ell/2$ from the boundary, i.e., $\hat{\vec w}\cdot \vec x - \hat b \in [-\ell/2, 0]$. Note that all such points are in the $\ell$-margin of the halfspace defined by $(\hat{\vec w}, \hat b + \ell/2)$. Additionally, these points are at distance of at least $\ell/2$  from this halfspace. Therefore, 
\begin{equation}
\label{eq:proof-soft-margin1.2b}
\left( \frac{1}{2} \right) \left( \frac \ell 2\right) \mden(\hat {\vec w}, \hat b)  \leq \sden(\hat {\vec w}, \hat b +\ell/2).    
\end{equation}

\end{enumerate}
The claim is proved using Equations~\ref{eq:proof-soft-margin1}, \ref{eq:proof-soft-margin1.2a}, and \ref{eq:proof-soft-margin1.2b}.

\subsection{Proof of Theorem~\ref{thm:main-g-threshold}}

For ease of exposition, let $\epsilon' = \min\{\epsilon^4, \epsilon^2 \sigma^4/r^4 \}$, $\nu = 1/ \exp\left(1/\epsilon \right)$, and $R = 2r + \sigma\sqrt{2\ln(1/\nu)}$.
Let's start with the reformulating the optimization problem as follows.
\begin{align}
\max_{g\in \Gth} \gain(g)= \max_{\substack{\vec w\in \mathrm{Img}(P)\\ b\in \R}} (\vec w\cdot \vec w_f) ~\sden_\D^{1/c}(\vec w, b) = \max_{\eta\in[0,1]} \eta  \max_{\substack{\vec w\in \mathrm{Img}(P)\\ \vec w\cdot \vec w_f = \eta \\ b\in \R}} \sden_\D^{1/c}(\vec w, b).
\label{eq:formulation-eta-01}
\end{align}
Let $(\vec w^*, b^*)$ be the optimizer of Equation~\eqref{eq:formulation-eta-01}.
Next, we show the value of solution $(\vec w^*, b^*)$ can be approximated well by a solution $(\hat{\vec w}, \hat b)$ where $\hat{\vec w}\cdot \vec w_f$ and $b$ belongs to a discrete set of values.
Let $N:=\{0, \epsilon' \|P \vec w_f\|, 2\epsilon' \|P \vec w_f\|\dots, \|P \vec w_f\| \}$. 
Let $\eta^* = \vec w^* \cdot \vec w_f$ and define $i$ such that 
$\eta^*\in (i \epsilon'\|P\vec w_f\|, (i+1) \epsilon' \|P\vec w_f\| ]$.
If $i> \lfloor \frac 1{\epsilon'} \rfloor - 2$, then let $\hat{\vec w} = \frac{P\vec w_f}{\| P \vec w_f\|}$. If  $i \geq \lfloor \frac 1{\epsilon'} \rfloor - 2$, then let $\hat{\vec w}$ be the unit vector defined by Lemma~\ref{lem:threshold:discretize} when $\vec w_1 = \vec w^*$ and $\vec w_2 = \frac{P\vec w_f}{\| P \vec w_f\|}$. Then,
\[
\hat{\vec w} \cdot \vec w_f = j \epsilon' \| P\vec w_f\|, \text{where} j = \begin{cases}
\lfloor\frac 1{\epsilon'} \rfloor & \text{If $i \geq \lfloor \frac 1{\epsilon '} \rfloor - 2$}\\
i + 1 & \text{Otherwise} 
\end{cases}.
\]
By these definitions and Lemma~\ref{lem:threshold:discretize}, we have that $\hat{\vec w} \cdot \vec w^* \geq 1- 2\epsilon' \| P \vec w_f\|$, which implies that
$\theta(\vec w^*, \hat{\vec w}) \leq \sqrt{2\epsilon' \|P \vec w_f\|}$.

Next, we use the lipschitzness of the soft-margin density to show that the soft-margin density of $\vec w^*$ and $\hat{\vec w}$ are close. Let $\hat b$ be the closest multiple of $\epsilon' \| P \vec w_f\|$ to $b$.
Using Lemma~\ref{lem:threshold:soft-margin-lipschitz}, we have that
\begin{align*}
\left| \sden_\D(\vec w^*, b^*) - \sden_\D(\hat{\vec w}, \hat b) \right| \leq O\left( \nu + \frac{\epsilon' \| P \vec w_f\|}{\sigma^2} R^2 \right) \in O\left(\nu + \frac{R^2}{\sigma^2} \sqrt{2\epsilon' \|P \vec w_f\|} \right).
\end{align*}
Note that for any $|a_1-a_2|\leq\epsilon'$ and $|b_1 - b_2| \leq \epsilon'$, $|a_1b_1 - a_2b_2| \leq (a_1 + b_1)\epsilon' + \epsilon'^2$, therefore,
\begin{equation}
\left(\hat{\vec w} \cdot \vec w_f\right) \sden_\D^\ell(\hat{\vec w}, \hat b) \geq 
\left(\vec w^* \cdot \vec w_f\right) \sden_\D^\ell(\vec w^*, b^*) -  O\left(\nu  + \frac{R^2}{\sigma^2} \sqrt{2\epsilon' \|P \vec w_f\|}   \right).
\label{eq:approx-grid}
\end{equation}
By Lemma~\ref{lem:threshold:1/4-approx}, the outcome of the algorithm has the property that
\[
\gain(\bar{\vec w}, \bar{b}) \geq \frac 14 \gain(\hat{\vec w}, \hat{b}) 
\]
Putting this all together, we have
\[
\gain(\bar{\vec w}, \bar{b}) 
\geq 
\max_{g\in \Gth}\gain(g) -   O\left(\nu  + \frac{R^2}{\sigma^2} \sqrt{2\epsilon' \|P \vec w_f\|}   \right)
\]

Now, replacing back $\epsilon' = \min\{\epsilon^4, \epsilon^2 \sigma^4/r^4 \}$, $\nu = 1/ \exp\left(1/\epsilon \right)$, and $R = 2r + \sigma\sqrt{2\ln(1/\nu)}$, we have
\[ \nu  + \frac{R^2}{\sigma^2} \sqrt{2\epsilon' \|P \vec w_f\|} \leq O\left( \epsilon + \frac{r^2}{\sigma^2} \sqrt{\epsilon'} + \frac{\sigma^2}{\sigma^2 \epsilon}  \sqrt{\epsilon'}\right) \leq O\left( \epsilon 
\right)
\]

\section{Proof of Theorem~\ref{thm:sample}}
\label{app:sample}
The first claim directly follows from Lemma~\ref{lem:pseudo-mech}, Theorem~\ref{thm:pollard} and the facts that $\gain_{\vec x}(g) \in [0, \ell]$ and  $\gain(g) = \E_{\vec x\sim\D}[\gain_{\vec x}(g)]$.

For the second claim, note that Algorithm~\ref{alg:opt-w-oracle} accesses distribution $\D$ only in steps that maximize the margin density of a classifier (the calls to oracle $\mathrm{O}(\D, \frac 1c, \K_\eta)$) or maximize the gain of a classifier. By Theorem~\ref{thm:pollard}, Lemma~\ref{lem:pseudo-mech}, and the fact that $\mden(g) = \E_{\vec x\sim\D}[\mden_{\vec x}(g)]$, we have that the density of any outcome of the oracle $\mathrm{O}(S, \frac 1c, \K_\eta)$ is optimal up to $O(\epsilon/c)$ for cost unit $c$.
Lemma~\ref{lem:threshold:1/4-approx} shows that the soft-margin density and the gain is also within a factor $O(\epsilon)$ from the $1/4$ approximation.

%% file: app_fullrank.tex
\section{Threshold Mechanisms when All Features are Visible}
\label{app:brendan}

Consider the special case where the projection matrix $P$ is the identity matrix (or, more generally, has full rank), and the mechanism is tasked with implementing a desired threshold rule.  In this case, we want to maximize the number of individuals whose features satisfy some arbitrary condition, and the features are fully visible to the mechanism.  Suppose that $f:\R^n \rightarrow \{0,1\}$ is an arbitrary binary quality metric, and the mechanism design space is over all functions $g:\R^n \rightarrow [0,1]$ (even non-binary functions).  Then in particular it is possible to set $g=f$, and we observe that this choice of $g$ is always optimal.

\begin{prop}
\label{prop:fullrank}
For any $f:\R^n \rightarrow \{0,1\}$, if $\G$ is the set of all functions $g:\R^n \rightarrow [0,1]$, then
$f \in \argmax_{g\in \G} \Val(g)$.
\end{prop}
\begin{proof}
Fix $f$, and write $S_f = \{ \vec x \in \R^n \ |\  f(\vec x) = 1 \}$ for the set of feature vectors selected by $f$.  Let $\bar{S}_f = \{ \vec x \in \R^n \ |\ \exists \vec y \text{ s.t. } f(\vec y) = 1, \cost(x,y) < 1/c \}$.  That is, $\bar{S}_f$ contains all points of $S_f$, plus all points that lie within distance $1/c$ of any point in $S_f$.  Write $\Phi_f = \Pr_{\vec x \sim \D}\big[ \vec x \in \tilde{S}_f \big]$.  Then note that $\Phi_f$ is an upper bound on $\Val(g) = \E_{\vec x \sim \D}\big[ f(\delta_g(\vec x)) \big]$ for any $g$, since $\Phi_f$ is the probability mass of all individuals who could possibly move their features to lie in $S_f$ at a total cost of at most 1, even if compelled to make such a move by a social planner, and $1$ is the maximum utility that can be gained by any move.

We now argue that setting $g=f$ achieves this bound of $\Phi_f$, which completes the proof.  Indeed, by definition, if $\vec x \in \bar{S}_f$ then there exists some $\vec y \in S_f$ such that $\cost(x,y) < 1$.  So in particular $U_f(\vec x, \vec y) > 0$, and hence $U_f(\vec x, \delta_f(\vec x)) > 0$.  We must therefore have $f(\delta_f(\vec x)) = 1$ for all $\vec x \in \bar{S}_f$, since $U_f(\vec x, \vec x') \leq 0$ whenever $f(\vec x') = 0$.  We conclude that $\Val(f) = \E_{\vec x \sim \D}\big[ f(\delta_f(\vec x)) \big] = \Phi_f$, as required.
\end{proof}

This proposition demonstrates that our problem is made interesting in scenarios where the projection matrix $P$ is not full-rank, so that some aspects of the feature space is hidden from the mechanism.

We note that Proposition~\ref{prop:fullrank} does not hold if we allow $f$ to be an arbitrary non-binary quality metric $f:\R^n \rightarrow [0,1]$.  For example, suppose $f(\vec x) = 0.1$ when $x_1 \geq 0$ and $f(\vec x) = 0$ otherwise, $c = 1$, and $\D$ is the uniform distribution over $[-1,1]^n$. Then setting $g = f$ results in all agents with $x_1 \geq -0.1$ contributing positively to the objective, since agents with $x_1 \in (-0.1, 0)$ will gain positive utility by shifting to $x'_1 = 0$.  However, if we instead set $g$ such that $g(\vec x) = 1$ when $x_1 \geq 0$ and $g(\vec x) = 0$ otherwise, then all agents will contribute positively to the objective, for a strictly greater aggregate quality.